\documentclass{article}

\usepackage[margin=3cm]{geometry}
\usepackage{amsmath,amsthm,amssymb}
\usepackage{cleveref}
\usepackage{xcolor}
\newcommand\A{\ensuremath{\mathcal{A}}}
\newcommand\stash{\ensuremath{s_\A}}
\newcommand\stashopt{\ensuremath{s_{\rm opt}}}
\newcommand\ov{\mathsf{ov}}
\newcommand\out{\mathsf{out}}
\newcommand\bigo{\mathcal{O}}
\newcommand\E{\mathcal{E}}

\newenvironment{subproof}[1][\proofname]{
	
	\begin{proof}[#1]
}{
	\end{proof}
}

\newtheorem{prop}{Proposition}
\newtheorem{claim}{Claim}
\newtheorem{lemma}{Lemma}
\newtheorem{theorem}{Theorem}
\newtheorem{corollary}{Corollary}

\usepackage{amssymb}

\begin{document}



\author{Brice Minaud$^1$, Charalampos Papamanthou$^2$}
\date{\small $^1$ Ecole Normale Sup\'{e}rieure, PSL University, CNRS, Inria, France\\
$^2$ Yale University, USA}



\title{Generalized Cuckoo Hashing with a Stash, Revisited}



\maketitle

\begin{abstract}
Cuckoo hashing is a common hashing technique, guaranteeing constant-time lookups in the worst case.
Adding a stash was proposed by Kirsch, Mitzenmacher, and Wieder at SICOMP 2010, as a way to reduce the probability of failure (i.e., the probability that a valid Cuckoo assignment fails to exist).
It has since become a standard technique in areas such as cryptography, where a negligible probability of failure is often required.
We focus on an extension of Cuckoo hashing that allows multiple items per bucket, which improves the load factor. That extension was also analyzed by Kirsch \emph{et al.} in the presence of a stash. In particular, letting $d$ be the number of items per bucket, and $s$ be the stash size, Kirsch \emph{et al.} showed that, for constant $d$ and $s$, the failure probability is $\bigo(n^{(s+1)(1-d)})$. In this paper, we first report a bug in the analysis by Kirsch \emph{et al.} by showing a counter-example leading to an asymptotically-larger probability of failure $\Omega(n^{-d-s-1})$. Then we provide a general analysis and upper bound of the failure probability for (almost) arbitrary $d$ and $s$, instead of just constant, which is useful for applications in cryptography. We finally deduce from the general analysis a tight bound $\Theta(n^{-d-s})$ for the probability of failure, for constants $d$ and $s$.

\end{abstract}






\section{Introduction}
\label{sec:intro}

Cuckoo hashing was introduced by Pagh and Rodler~\cite{cuckoo}, and proceeds as follows. We wish to allocate a set $S$ of $n$ items into two tables $T_1$, $T_2$ of size $m = (1+\varepsilon)n$ each, where $\varepsilon > 0$ is an arbitrary constant.
The construction is parametrized by two uniformly random hash functions $h_1: S \to T_1$ and $h_2: S \to T_2$.
Each item $x \in S$ may be allocated either to $h_1(x)$, or to $h_2(x)$.
Pagh and Rodler prove that a valid Cuckoo assignment (where no two items are assigned to the same location) exists with probability $1-\bigo(n^{-1})$ over the randomness of the hash functions. Cuckoo hashing guarantees constant-time lookups: Any item $x$ can be retrieved by visiting two memory locations. This makes Cuckoo hashing attractive in real-time systems, where worst-case performance is critical as well as in cryptography, where it serves as a core component for certain oblivious algorithms \cite{obliviousHash}. 
Two variations of the above simple Cuckoo hashing construction have been considered in the literature, which are relevant to our paper.

\paragraph{Large buckets for improved load factor}  Dietzfelbinger and Weidling~\cite{balanced} consider a simple tweak of Cuckoo hashing, where up to $d>1$ items can be assigned to the same location in the table. In that setting, there is only a single table $T$ consisting of $m = (1+\varepsilon)n/d$ buckets, each of size $d$. With this tweak, the load factor (ratio of occupied space in the table after allocating all items) improves to $1/(1+\varepsilon)$ and therefore can be made arbitrarily close to the optimum 1 (In the basic construction the load factor is bounded by $1/2$.)

 \paragraph{Adding a stash to handle failures} Kirsch, Mitzenmacher, and Wieder~\cite{kmw} study Cuckoo hashing in the presence of a \emph{stash}. The stash is used to store items that the Cuckoo insertion algorithm would otherwise fail to assign (and thus trigger a rehash). They show that when allowing for a stash of constant size $s$, the probability of failure of Cuckoo hashing becomes $\bigo(n^{-s})$. If $s$ is allowed to vary, a similar bound $\bigo(n^{-s/2})$ was proven in \cite{aumuller}, assuming $s = \bigo(n^{c})$ for sufficiently small constant $c$. In either case, the failure probability decreases exponentially with $s$. Note that having a stash is crucial for both real-time and cryptographic applications, where rehashing can be undesirable. In real-time systems, the rehash approach can be unsatisfactory, since it offers poor worst-case performance guarantees. In some cryptographic applications, rehashes are forbidden entirely, because security proofs rely on the fact that the hash functions are independent of the data being hashed, and rehashing breaks that property~\cite{alibi}\footnote{The issue is subtle: If rehashes are allowed, the hash functions are no longer uniformly random; instead, they are uniformly random conditioned on the event that a rehash is not necessary, and that event depends on the data being hashed.}.
 
\paragraph{Cuckoo graph}
For the formal treatment of our results, we use the well-established notion of a \emph{Cuckoo graph}: A Cuckoo hashing instance of a set $S$ of $n$ items stored in a table $T$ of $n(1+\varepsilon)/d$ buckets of size $d$ each in the presence of a stash of size $s$ is associated with the \emph{Cuckoo graph} $G(n,\varepsilon,d,s)$, with $n(1+\varepsilon)/d$ vertices and $n$ edges $$E = \{(h_1(x),h_2(x)) : x \in S\}\,.$$ This graph is undirected, may contain multiple copies of the same edge ($E$ is a multiset), and may contain loops.
Computing a valid Cuckoo hashing can be translated into operations on the Cuckoo graph, with the correspondence of Table~\ref{tab}. If the stash is allowed to be of size up to $s$, the existence of a valid Cuckoo assignment for hash functions $h_1$ and $h_2$ can be translated into purely graph-theoretic terms (Specifically, it translates to a graph orientability problem with deletions \cite{orientDeletion}.)
We say that the Cuckoo graph is \emph{suitable} if and only if it is possible to remove at most $s$ edges, and orient the remaining edges, such that every vertex has outdegree at most $d$.
We will refer to the \emph{failure probability} as the probability that $G(n,\varepsilon,d,s)$ is not suitable, over the choice of the hash functions $h_1$ and $h_2$.

\begin{table}[bh]
\caption{From Cuckoo hashing to Cuckoo graph.\label{tab} }
\begin{center}
\begin{tabular}{|l|l|}
\hline
\textsf{Cuckoo hashing} & \textsf{Cuckoo graph}\\\hline
Assigning item $x$ to cell $h_1(x)$ & Orienting edge $(h_1(x),h_2(x))$ towards $h_2(x)$\\
Assigning item $x$ to cell $h_2(x)$ & Orienting edge $(h_1(x),h_2(x))$ towards $h_1(x)$\\
Assigning item $x$ to the stash & Removing edge $(h_1(x),h_2(x))$\\
Number of items stored in cell $c$ & Outdegree of vertex $c$\\
\hline
\end{tabular}
\end{center}
\end{table}

\paragraph{Our starting point} In this work, we are interested in studying the failure probability of Cuckoo hashing in the presence of a stash and when buckets of size $d>1$ are used. Kirsch, Mitzenmacher, and Wieder~\cite{kmw} conclude their seminal paper by claiming the following (We paraphrase here to adjust to our notation.)
\begin{prop}[\cite{kmw} Proposition 4.3]
For any constants $\varepsilon > 0$, $d \geq 1 + \ln(1/\varepsilon)/(1-\ln 2)$, and $s \geq 0$, the probability that $G(n,\varepsilon,d,s)$ is not suitable is $\bigo(n^{(s+1)(1-d)})$.
\label{prop:orig}
\end{prop}
Our starting point in this work is showing Proposition 4.3 from~\cite{kmw} cannot hold, by considering a ``bad event" (i.e., one that causes $G(n,\varepsilon,d,s)$ not to be suitable) that has probability $\Omega(n^{-d-s-1})$.
Define $\E_v$ to be the event that there exists a subset of items $X \subseteq S$ of cardinality $d + s + 1$, such that $\forall x \in X, h_1(x) = h_2(x) = v$. Clearly, if $\E_v$ holds, there can be no valid Cuckoo assignment, since the bucket $v$ can hold at most $d$ items, and at most $s$ can be relocated to the stash, and therefore $G(n,\varepsilon,d,s)$ is not suitable. So it is enough to compute a lower bound on $\E_v$, which will also be a lower bound for the failure probability. We give our result directly here.
(A tighter bound will be proved in \Cref{sec:constants}.)

\begin{lemma}[Lower bound of failure probability]\label{lower_lem}
Assume $\varepsilon > 0$, $d \geq 1$ and $s \geq 0$ are constant. The probability that $G(n,\varepsilon,d,s)$ is not suitable is $\Omega(n^{-d-s-1})$.
\end{lemma}
\begin{proof}
Since the hash functions are uniform, the probability that a given $x \in S$ satisfies $h_1(x) = h_2(x) = v$ is $m^{-2}$. 
$\E_v$ may be viewed as the event that there are at least $d+s+1$ successes in a binomial experiment with $n$ trials, each with probability of success $m^{-2} = 1/c^2n^2$ (where $c=(1+\epsilon)/d$ is a constant). Therefore
\begin{align*}
\Pr(\E_v) &= \sum_{k \geq d+s+1}{n \choose k}m^{-2k}(1-m^{-2})^{n-k}\\& \geq (1-m^{-2})^n \sum_{k \geq d+s+1}\left(\frac{n}{k}\right)^k m^{-2k}\\
&=(1-1/(c^2n^2))^n \sum_{k \geq d+s+1}\left(\frac{n}{kc^2n^2}\right)^k\\&\geq (1-1/(c^2n^2))^n \cdot \left(\frac{1}{(d+s+1)c^2n}\right)^{d+s+1}\\&=\Omega(n^{-d-s-1})\,,
\end{align*}
since $(1-1/(c^2n^2))^n \to_{n \to \infty} 1$.
\end{proof}

\noindent \textbf{Origin of the flaw.}
Lemma~\ref{lower_lem} implies that Proposition~4.3 from \cite{kmw} is flawed.
Looking at the proof of Proposition~4.3 in \cite{kmw}, the issue stems from the expression $F \leq \sum_{s+1 \leq j \leq m/(1+\varepsilon)} F(j)$, near the beginning of the proof. The sum should start from $j = 1$, not $j = s+1$. This has a large impact on the final result, because the $F(j)$'s are (roughly) exponentially decreasing; in particular, $F(1)$ dominates the sum.

The argument in \cite{kmw} is, roughly speaking, that each term $F(j)$ can be upper-bounded by $O(n^{j(1-d)})$. Because the authors assume $j\geq s+1$, they get a bound of the form $O(n^{(s+1)(1-d)})$ for the overall sum. Since in reality, the sum must start at $1$, this line of reasoning can only yield a bound $O(n^{1-d})$, which does not depend on the stash size $s$, defeating the point of the analysis.

For that reason, we set out in this article to repair Proposition~4.3, using a different proof technique. Our analysis follows the same structure as the one in $\cite{balanced}$, which previously studied the same problem without a stash. In the end, the new analysis mainly comes down to a finer upper bound for the terms $F(j)$'s, that properly accounts for the presence of a stash.

\noindent \textbf{Our main results.}
Our main results are organized as follows.
\begin{itemize}
\item In \Cref{sec:proof}, we repair Proposition~4.3 from \cite{kmw}, as explained above. In particular we compute an upper bound on the probability of $G(n,\varepsilon,d,s)$  not being suitable for \emph{(almost) arbitrary values} of $d$ and $s$, instead of just constants as in~\cite{kmw}. The only requirements we have is that  $d \geq 1 + \ln(1/\varepsilon)/(1-\ln 2)$ (this condition is tight by \cite[Proposition~4.4]{kmw}) and $1 \leq s \leq m/(10e^4)$. We show the failure probability is 
\[
\bigo\left( \left(\frac{e}{m} \right)^{d+s} + \left(\frac{2e}{m}\right)^{2(d-1)} \left(\frac{\alpha s}{m}\right)^{s+1}  + m\gamma^m \right)\,,
\]
for some constants $\alpha > 0$ and $\gamma < 1$. See Theorem~\ref{prop:up}. Note that the above is negligible in $n$ when $s$ and $d$ are, for example, $\log n$, which is important for applications in cryptography (e.g.,~\cite{DBLP:conf/crypto/BossuatBFMR21}). 

\item In Section~\ref{sec:constants}, we focus on the case that $d$ and $s$ are constants, and give a tight bound $\Theta(n^{-d-s})$ for the probability that $G(n,\varepsilon,d,s)$ is not suitable. We leave as open problem to derive a tight bound for arbitrary $d$ and $s$.
\item Finally, in Section~\ref{algo_result}, we examine how our updated result affects the guarantees of the algorithm that actually does the final assignment (and which was presented in~\cite{kmw}), showing that the assignment algorithm from~\cite{kmw} will output a stash whose size exceeds $s$ with probability $\bigo(n^{-d-s})$---see Corollary~\ref{final_col}.   

\end{itemize}

Throughout this paper, $\varepsilon$ is viewed as a constant, as is standard in the analysis of Cuckoo hashing schemes. In other words, the failure probability is regarded as a function of $n$, $d$, $s$, for fixed $\varepsilon > 0$. Concretely, this means that the hidden constants in the $\bigo()$, $\Omega()$, $\Theta()$ notation may depend on $\varepsilon$.
In~\cite{kmw}, $d$ and $s$ are also regarded as constants, with the same implication.
Parts of our analysis (e.g., in Section~\ref{sec:constants}) will also view $d$ and $s$ as constant.
Which quantities are constant will be explicit in theorem statements.

\section{Failure probability upper bound for general $d$ and $s$}
\label{sec:proof}
In this section, we repair Proposition 4.3 from~\cite{kmw} and prove a generalized result for the failure probability. We stress that unlike Proposition 4.3 from~\cite{kmw}, we do not require $d$ and $s$ to be constants.

In what follows, it will be convenient to use both cuckoo hashing terminology, and cuckoo graph terminology, via the correspondance discussed in the introduction, depending on the situation. In particular, we will sometimes identify an item $x \in S$ with the corresponding edge $(h_1(x),h_2(x)) \in E$. (This is a one-to-one mapping: Even if two items $x \neq y \in S$ are such that $h_1(x) = h_1(y)$ and $h_2(x) = h_2(y)$, recall that $E$ is a multiset, and each item gives rise to a distinct edge.)

Our proof follows the same approach as \cite[Proof of Theorem~1]{balanced}, but has to account for the addition of a stash. We first introduce some necessary notation: For $X \subseteq E$ being a subset of edges, let $\Gamma(X)$ be the set of endpoints of edges in $X$ (E.g., if $X=\{(1,2),(2,3),(1,3),(1,1)\}$ then $\Gamma(X)=\{1,2,3\}$.) Without a stash (\emph{i.e.,} $s=0$), \cite{balanced} observes that $G(n,\varepsilon,d,s)$  is suitable iff $\forall X \subseteq S, |\Gamma(X)| \geq |X|/d$. If the condition fails, \emph{i.e.,} $\exists X, |\Gamma(X)| < |X|/d$, suitability fails ``trivially'' because there is not enough room to store the $|X|$ edges in $\Gamma(X)$. So the previous equivalence may be understood as saying: suitability holds iff it does not fail trivially on any subset of edges. This is still the case when adding a stash. We present the following claim.

\begin{lemma}
$G(n,\varepsilon,d,s)$ is suitable iff for all $X \subseteq S$, $d|\Gamma(X)| + s \geq |X|$.
\label{fact:key}
\end{lemma}

\begin{proof}
If the condition $d|\Gamma(X)| + s \geq |X|$ fails for some subset $X$ of edges, then there is not enough room in $\Gamma(X)$ and the stash to store the edges in $X$.
Indeed, each vertex in $|\Gamma(X)|$ can hold at most $d$ items, and the stash can hold at most $s$ items.
This shows that if $G$ is suitable, then the condition must hold for all $X$. Conversely, assume the graph is not suitable. We are going to build an $X$ such that the condition fails.

Define the \emph{overflow} $\ov(D)$ of a directed graph $D$ as the minimal number of edges to remove such that every vertex has outdegree at most $d$. Let $s'$ be the smallest integer such that the following statement holds: there exists a directed graph $D$ arising from orienting the edges of $G$ such that $\ov(D) = s'$ (in other words, $s'$ is the minimum stash size, across all possible orientations of $G$). The fact that $G$ is not suitable translates to $s' > s$. Fix $D$ witnessing the previous statement. Consider the subset $Y$ of vertices of $D$ that have outdegree strictly more than $d$. Let $Y' \supseteq Y$ be the set of vertices that can be reached from $Y$ by following a directed path. Observe that the outdegree of every vertex $v$ in $Y'$ must be at least $d$. Otherwise, there would exist a directed path from a vertex $v$ in $Y$ to some $w$ in $Y'$ with outdegree $< d$. Flipping the direction of every edge along this path decreases the outdegree of $v$ by 1, increases the outdegree of $w$ by 1, and does not change the outdegree of intermediate vertices. Because $v$ was over capacity (outdegree $> d$) and $w$ was under capacity (outdegree $< d$), this means that after flipping the edges of the path, the overflow has decreased by 1. This would contradict the minimality of $s'$. Let $D' = (Y',X)$ be the subgraph of $D$ induced by $Y'$. Here, $X$ is the set of edges of $D$ whose endpoints are both in $Y'$. Note that $Y' = \Gamma(X)$. By construction, we have $s' = \ov(D)  = \ov(D') = \sum_{v \in Y'} (\out_{D'}(v) - d) = |X| - d|Y'| = |X| - d|\Gamma(X)|$. Since $s' > s$, $X$ witnesses $d|\Gamma(X)| + s < |X|$.
\end{proof}

Lemma~\ref{fact:key} will be useful in proving the main theorem, presented next.

\begin{theorem}\label{prop:up}
Fix a constant $0 < \varepsilon \leq 0.25$. There exist constants $\alpha > 0$ and $\gamma < 1$, such that for every $d \geq 1 + \ln(1/\varepsilon)/(1-\ln 2)$, and $1 \leq s \leq m/(10e^4)$, the probability that $G(n,\varepsilon,d,s)$ is not suitable is
\[
\bigo\left( \left(\frac{e}{m} \right)^{d+s} + \left(\frac{2e}{m}\right)^{2(d-1)} \left(\frac{\alpha s}{m}\right)^{s+1}  + m\gamma^m \right).
\]
If $d$ and $s$ are constants, that expression is $\bigo(n^{-d-s})$.
\end{theorem}



\begin{proof}




By Lemma~\ref{fact:key}, the probability that $G$ is not suitable is equal to
\[
F := \Pr(\exists X \subseteq S: d|\Gamma(X)| + s < |X|)\,.
\]
Now, let $F(j)$ be the probability that there exists a set $Y$ of $j$ vertices and a set $X$ of $dj + s + 1$ edges that satisfy $\Gamma(X) \subseteq Y$.
Note that $F(j)$ must be zero when $dj+s+1 > n$, since $X$ must be a subset of $S$ with $|S| = n$. Therefore, letting $J=(n-s-1)/d$, by the union bound we have
\begin{equation}\label{total_sum}
F \leq \sum_{1\leq j \leq J} F(j)\,,
\end{equation}
Note also that 
\begin{equation}\label{binomial}
F(j)\le {m \choose j} \Pr[I_j\ge dj+s+1]\,,
\end{equation}
where $I_j$ is the number of edges from $S$ whose endpoints both fall into a fixed set $Y$ of size $j$. Since $I_j$ follows a binomial distribution with $\mathbb{E}[I_j]=n(j/m)^2$, we can use the Chernoff–Hoeffding bound, as in~\cite{balanced}, to bound $\Pr[I_j\ge jd+s+1]$, and therefore for $j < J$, Eq.~(13) from \cite{balanced}, in the presence of the stash, becomes
\begin{align}
F(j) \leq {m \choose j}\left( \frac{nj}{m^2(d+s'/j)} \right)^{jd+s'}\left( \frac{n(m^2-j^2)}{m^2(n-jd-s')} \right)^{n-jd-s'}\,,
\tag{13'}\label{eq:13}
\end{align}
where we set for convienience $s'=s+1$ (Some equation indices follow \cite{balanced} when possible for ease of comparison, adding a prime, so (13') is the counterpart of (13) in \cite{balanced}.) Continuing from there and replacing $n$ by $dm/(1+\varepsilon)$ we have
\begin{align}
F(j) &\leq \frac{m^m}{j^j(m-j)^{m-j}} \left( \frac{d}{d+s'/j} \cdot \frac{j}{(1+\varepsilon)m} \right)^{jd+s'}\left( \frac{d(m^2-j^2)}{(1+\varepsilon)m(n-jd-s')} \right)^{dm/(1+\varepsilon)-jd-s'}\notag\\
&= \frac{m^m}{j^j(m-j)^{m-j}} \left( \frac{d}{d+s'/j} \cdot \frac{j}{(1+\varepsilon)m} \right)^{jd+s'}\left( \frac{m^2-j^2}{(1+\varepsilon)m(m/(1+\varepsilon)-j-s'/d)} \right)^{dm/(1+\varepsilon)-jd-s'}\notag\\
&< (1+\varepsilon)^{-jd-s'}m^{m\frac{1+\varepsilon-d}{1+\varepsilon}}j^{j(d-1)+s'}(m-j)^{j-m}\left( \frac{m^2-j^2}{m-(1+\varepsilon)(j+s'/d)} \right)^{d\frac{m-(j+s'/d)(1+\varepsilon)}{1+\varepsilon}}\nonumber \\
&=f(j,\varepsilon)\,.
\tag{14'}\label{eq:14}
\end{align}
We split the sum from Equation~\ref{total_sum} into three cases, $j = 1$, $1 < j < 5s$, and $5s \leq j \leq J$.

\bigskip
\noindent {\bf Case 1:} $j = 1$. By (\ref{eq:13}) we get:
\begin{align}
F(1) &\leq m\left( \frac{d}{d+s'} \cdot \frac{1}{(1+\varepsilon)m} \right)^{d+s'}\left( \frac{n-n/m^2}{n-d-s'} \right)^{n-d-s'}\notag\\
&< m\left( \frac{1}{(1+\varepsilon)m} \right)^{d+s'}e^{d+s'-n/m^2}\notag\\
&= \bigo\left( \left( \frac{e}{m} \right)^{d+s} \right).
\tag{15'}\label{eq:15}
\end{align}
Note that if $s$ and $d$ are constant, this is $\bigo(n^{-d-s})$.

\bigskip
\noindent {\bf Case 2:} $1 < j \le 5s$. Recall that the upper-bound $F(j) \leq f(j,\varepsilon)$ in Equation~(\ref{eq:14}) was computed as a union bound as follows 
\[
F(j) \leq {m \choose j}  \Pr[I_j\ge dj+s+1] \leq f(j,\varepsilon)\,.
\]
It is clear that $\Pr[I_j\ge dj+s+1]$, a function of $\varepsilon$, is decreasing while $\varepsilon$ is increasing because $\Pr[I_j\ge dj+s+1]$ is the upper tail of a binomial experiment, and $\varepsilon$ only decreases the probability of success of each trial in the experiment, while it does not affect any other parameter. It follows that $F(j)  \leq f(j,0)$. Using (\ref{eq:14}) we get
\begin{align}
F(j) &< f(j,0) = m^{m(1-d)}j^{j(d-1)+s'}(m-j)^{j-m}\left( \frac{m^2-j^2}{m-j-s'/d} \right)^{d(m-j-s'/d)}\notag\\
&= m^{m(1-d)}j^{j(d-1)+s'}m^{j-m}(1-j/m)^{j-m} m^{d(m-j-s'/d)} \cdot \left( 1+\frac{j}{m}+\frac{s'}{dm}\cdot\frac{1+j/m}{1-j/m-s'/(dm)} \right)^{d(m-j-s'/d)}\notag\\
&< m^{j(1-d)-s'}j^{j(d-1)+s'}e^{j(m-j)/m}e^{d(m-j-s'/d)\left(j/m+\frac{s'}{dm}\cdot\frac{1+j/m}{1-j/m-s'/(dm)}\right)}\notag\\
&= \left( \frac{j}{m} \right)^{(d-1)j+s'}e^{j(m-j)/m}e^{d(m-j-s'/d)\left(j/m+\frac{s'}{dm}\cdot\frac{1+j/m}{1-j/m-s'/(dm)}\right)}\notag\\
&= \left( \frac{j}{m} \right)^{(d-1)j+s'}e^{j(m-j)/m}e^{dj(m-j-s'/d)/m}e^{d(m-j-s'/d) \cdot \left(\frac{s'}{dm}\cdot\frac{1+j/m}{1-j/m-s'/(dm)}\right)}\notag\\
&< \left( \frac{j}{m} \right)^{(d-1)j+s'}e^{(d+1)j}e^{s'\left( \frac{1+j/m}{1-j/m-s'/(dm)}\right)}\notag.
\end{align}
Because $j < 5s$, $s \leq m/(10e^4)$ and $d$ is bounded below by a constant, $e^{(1+j/m)/(1-j/m-s'/(dm))}$ is upper-bounded by an absolute constant $C$. Reinjecting $C$ we have:
\begin{align}
F(j) &< \left( \frac{j}{m} \right)^{(d-1)j+s'}e^{(d+1)j}C^{s'}\notag\\
&= \left( \left(\frac{j}{m} \right)^{d-1}e^{d+1}\right)^j\left(\frac{Cj}{m}\right)^{s'}.
\tag{17'}\label{eq:17}
\end{align}

In the end, the expression is the same as in \cite{balanced}, with an additional term $(Cj/m)^{s'}$. Denote the rest of the expression by $g(j)$, so that (\ref{eq:17}) becomes $F(j) < g(j)(Cj/m)^{s'}$. It was already observed in \cite{balanced} that in the relevant range $1 < j \leq 5s \leq m/(2e^4)$, $g(j)$ is geometrically decreasing. More precisely, the ratio of one term to the previous term is upper-bounded by a constant strictly less than $1$. It follows that $g(2)$ dominates $\sum_{2 \leq j < 5s} g(j)$.
This yields
\begin{align*}
\sum_{2 \leq j < 5s} F(j) = \bigo\left( g(2) \left(\frac{C\cdot 5s}{m}\right)^{s'} \right)
= \bigo\left( \left(\frac{2e}{m}\right)^{2(d-1)} \left(\frac{\alpha s}{m}\right)^{s+1} \right)\,,
\end{align*}
for some constant $\alpha>0$. Note that if $s$ and $d$ are constant, this is $\bigo(n^{1-2d-s})=\bigo(n^{-d-s})$.

\bigskip

\noindent {\bf Case 3:} $5s < j \leq J=(n-s-1)/d$. The case $5s<j\le m/(2e^4)$ is handled in Case 2. The case $j>m/(2e^4)$ was already shown to be $\bigo(m\gamma^m)$ for some $\gamma < 1$ in \cite{balanced}. Note that if $s$ and $d$ are constant, this is $\bigo(n^{-d-s})$. Therefore overall, if $s$ and $d$ are constant, the probability of failure is $\bigo(n^{-d-s})$. This concludes the proof. 
\end{proof}

\section{Failure probability tight bound for constant $d$ and $s$}\label{sec:constants}
In this section we prove a tight bound for the failure probability, for constant $d$ and $s$. 
\begin{theorem}\label{prop:up2}
Fix a constant $0 < \varepsilon \leq 0.25$. For every constants $d \geq 1 + \ln(1/\varepsilon)/(1-\ln 2)$ and $s\ge 1$, the probability that $G(n,\varepsilon,d,s)$ is not suitable is $\Theta(n^{-d-s})$.
\end{theorem}
\begin{proof}
By Theorem~\ref{prop:up}, we know that for constants $d$ and $s$, the probability of failure is $\bigo(n^{-d-s})$. It is enough to show that the failure probability is $\Omega(n^{-d-s})$. To do that, we will first show that $\Pr(\mathcal{E}_v)=\bigo(n^{-d-s-1})$, where $\mathcal{E}_v$ is the event defined in the introduction. Indeed,
\begin{align*}
\Pr(\E_v) &= \sum_{k \geq d+s+1}{n \choose k}m^{-2k}(1-m^{-2})^{n-k}\\ 
&\leq \sum_{k \geq d+s+1}\frac{n^k}{k!} m^{-2k}\\
&= \sum_{k \geq d+s+1}\frac{n^{-k}}{k!} \left(\frac{1+\varepsilon}{d}\right)^{-2k}\\
&\leq n^{-d-s-1}\sum_{k \geq d+s+1}\frac{1}{k!} \left(\frac{1+\varepsilon}{d}\right)^{-2k}\\
&\leq n^{-d-s-1}e^{\left(\frac{1+\varepsilon}{d}\right)^{-2}}\\
&= \bigo(n^{-d-s-1}).
\end{align*}
Therefore, by Lemma~\ref{lower_lem}, we have $\Pr(\mathcal{E}_v)=\Theta(n^{-d-s-1})$. We continue by showing that the failure probability is bounded from below by $m\Pr(\mathcal{E}_v)/2$, thus completing the proof. 

Indeed, consider the binary random variable $I_v$ that is equal to $0$ iff $\E_v$ occurs. Likewise, let $I$ denote the binary random variable that is equal to $0$ iff $\E = \bigvee_v \E_v$ occurs. Observe $I = \prod_v I_v$.
\begin{claim}\label{neg}
Variables $I_v$ are negatively associated.
\end{claim}
\begin{subproof} For $v \in V$, let $B_v = |\{x \in S : h_1(x) = h_2(x) = v\}|$ denote the number of edges that loop in on $v$ in the Cuckoo graph. Let $B' = |\{x \in S : h_1(x) \neq h_2(x)\}|$ denote the number of edges that are not loops, so that $B' + \sum_v B_v = n$.
By~\cite[Theorem~13]{neg}, the $B_v$'s together with $B'$ are negatively associated.
A fortiori the $B_v$'s are negatively associated. We have $I_v = f(B_v)$ where $f$ is the non-increasing function that is equal to $1$ if its input is strictly less than $d+s+1$, and $0$ otherwise.
Hence by~\cite[Proposition~7.2]{neg}, random variables $I_v$ are negatively associated. 
\end{subproof}

By~\cite[Proposition~3]{neg} and Claim~\ref{neg},  we have $\mathbb{E}(\prod_v I_v) \leq \prod_v \mathbb{E}(I_v)$. We therefore get
\begin{flalign*}
&&\Pr(\E) &= 1 - \mathbb{E}(I)&&\\
&&&= 1 - \mathbb{E}\Big(\prod_v I_v\Big)&&\\
&&&\geq 1 - \prod_v \mathbb{E}(I_v)&&\\
&&&= 1 - (1-\Pr(\E_v))^m&&\text{for any fixed $v$}\\
&&&\geq 1 - e^{-m\Pr(\E_v)}&&\text{using $1-x \leq e^{-x}$.}
\end{flalign*}
We have already established that $m\Pr(\E_v)=\Theta(n^{-d-s})$ (since $\Pr(\mathcal{E}_v)=\Theta(n^{-d-s-1})$), so we can freely assume $m\Pr(\E_v) < 1$ for $n\ge n_0$ where $n_0$ is a fixed constant. Using the fact that $1-e^{-x} > x/2$ for $x$ in $[0,1]$, we get
\[
\Pr(\E) \geq m\Pr(\E_v)/2 = \Omega(n^{-d-s})\,,
\]
which concludes the proof, since the the failure probability is bounded from below by $\Pr(\mathcal{E})$.
\end{proof}

\section{Algorithmic results\label{algo_result}}

In the previous sections, we have identified a mistake in \cite[Proposition~4.3]{kmw}, and repaired the proposition by proving a new result, with a corrected bound (\Cref{prop:up}). \cite[Proposition~4.3]{kmw} is later used to prove \cite[Theorem~4.2]{kmw}. As a result, \cite[Theorem~4.2]{kmw} becomes invalid\textemdash both the proof, and the bound in the theorem statement are incorrect. In this section, we repair \cite[Theorem~4.2]{kmw}. No other result in \cite{kmw} depends on \cite[Proposition~4.3]{kmw}, and to the best of our knowldege, the rest of the article is correct.

\Cref{prop:up} is a statement about the \emph{existence} of a solution to a problem. It provides an upper bound on the probability that a certain allocation problem fails to have a solution.
A natural question is how to find such a solution algorithmically. Kirsch \emph{et al.} propose an algorithm for that purpose. Using the terminology from \Cref{sec:intro}, the algorithm can be formulated as follows.

We start from a graph with vertex set $V = T_1 \cup T_2$, and no edge. Each item $x \in S$ is inserted in turn. To insert item $x$, the directed edge $(h_1(x),h_2(x))$ is added to the graph. The algorithm then looks for a directed path starting from $h_1(x)$, and ending in any vertex with outdegree strictly less than $d$. Such a path is called an \emph{augmenting path}. To find an augmenting path, the algorithm performs a breadth-first search starting from $h_1(x)$. If an augmenting path is found, the direction of every edge along the path is flipped. If the breadth-first search exhausts the vertices reachable from $h_1(x)$ without finding an augmenting path, the item $x$ is added to the stash; in that case, the edge $(h_1(x),h_2(x))$ is removed from the graph. Observe that this insertion algorithm preserves the invariant that every vertex has outdegree at most $d$. As a consequence, the graph $G$ output by the overall algorithm sastisfies $\Delta^+(G) \leq d$. Let us call that algorithm $\mathcal{A}$.

As noted in \cite{kmw}, results about the insertion time of $\mathcal{A}$ can be deduced directly from the analysis in \cite{balanced}. For that reason, Kirsch \emph{et al.} focus their analysis of $\mathcal{A}$ on the distribution of the stash size (\emph{i.e.} the distribution of the number of items that are sent to the stash by $\A$). The following result is claimed.

\begin{theorem}[\cite{kmw}, Theorem~4.2]
\label{thm:alg}
Les $\stash$ denote the size of the stash after all $n$ items have been inserted using algorithm $\A$. For small enough $\varepsilon > 0$ and $d \geq 15.8\cdot \ln(1/\varepsilon)$, and any integer $s \geq 2$, we have $\Pr(\stash \geq s) = O(n^{sd(D-1)})$, where $D = d^{-1/3} + d^{-1}$. For $s = 1$, we have $\Pr(\stash \geq s) = O(n^{1-d})$.
\end{theorem}

As noted earlier, the proof of \Cref{thm:alg} relies on \cite[Proposition~4.3]{kmw}, which we have shown to be false. In fact, the bound $O(n^{sd(D-1)})$ given in the theorem cannot hold, because it is supposed to bound to the probability that the algorithm fails to find a solution; but after correcting \cite[Proposition~4.3]{kmw}, it is now asymptotically lower than the probability $\Theta(n^{-d-s})$ that a solution exists in the first place (\Cref{prop:up2}).

To repair Theorem~4.2, we use a different approach from \cite{kmw}: We show that algorithm $\A$ is \emph{optimal}, in the sense that it minimizes the number of items sent to the stash. More precisely, given any fixed pair of hash functions $h_1, h_2$ (to avoid cluttering notation, we leave the hash functions as implicit parameters), let $\stash$ denote the size of the stash at the outcome of algorithm $\A$. On the other hand, let $\stashopt$ denote the smallest stash size such that a solution exists; formally, using the notation from \Cref{sec:intro}:
\[
\stashopt = \min \{ s : G(n,\varepsilon,d,s)\text{ is suitable}\}.
\]

\begin{theorem}
For any any pair of hash functions $(h_1, h_2)$, it holds that $\stash = \stashopt$.
\label{thm:opt}
\end{theorem}

\begin{proof}
Since $\A$ outputs a valid solution to the cuckoo allocation problem, $\stash \geq \stashopt$ holds trivially.
We now prove the converse.

First, let us introduce some notation. Given a directed graph $G = (V,E)$, let $G^u = (V,E^u)$ denote the undirected graph obtained by forgetting edge orientations in $G$.
Given a graph $G = (V,E)$, and $V' \subseteq V$ a subset of vertices, let $G[V'] = (V',E[V'])$ denote the subgraph of $G$ induced by $V'$, that is: $E[V'] = \{(u,v) \in E : u,v \in V'\}$.
Finally, a vertex $v$ in a directed graph $G$ is said to be \emph{full} if there is no augmenting path from $v$, \emph{i.e.} no directed path from $v$ to any vertex $w$ with $\deg^+(w) < d$. (Paths of length zero are allowed, hence $v$ being full implies $\deg^+(v) \geq d$.) The key lemma is the following.

\begin{lemma}
\label{fact:key2}
Let $G = (V,E)$ be a directed graph such that $\Delta^+(G) \leq d$, and let $V' \subseteq V$. The following three properties are equivalent.
\begin{itemize}
\item[(1)] For all $v \in V'$, $v$ is full in G.
\item[(2)] There exists $V'' \supseteq V'$ such that $|E[V'']| \geq d |V''|$.
\item[(3)] There exists $V'' \supseteq V'$ such that $|E[V'']| = d |V''|$.
\end{itemize}
\end{lemma}

\begin{subproof}
$(2) \Rightarrow (1)$. Assume $(2)$.
We have:
\[
d |V''| \leq |E[V'']| = \sum_{v \in V''} \deg^+_{G[V'']}(v) \leq \sum_{v \in V''} \deg^+_G(v) \leq \Delta^+(G)|V''| \leq d|V''|.
\]
Hence all inequalities are equalities, which implies $\deg^+_{G[V'']}(v) = \deg^+_G(v) = d$ for all $v \in V''$. Hence, there is no edge from $V''$ to $V \setminus V''$, and all vertices in $V''$ have outdegree $d$. Hence all vertices in $V'$ are full.

$(1) \Rightarrow (3)$.  Assume $(1)$. Let $V''$ be the set of vertices reachable from $V'$ in $G$. All vertices in $V''$ must have outdegree at least $d$, since the vertices of $V'$ are full. Since $\Delta^+(G) \leq d$, all vertices in $V''$ have outdegree exactly $d$. On the other hand, by construction, there is no edge from $V''$ to $V \setminus V''$. In conclusion:
\[
|E[V'']| = \sum_{v \in V''} \deg^+_{G[V'']}(v) = \sum_{v \in V''} \deg^+_G(v) = d|V''|.
\]

$(3) \Rightarrow (2)$ is trivial. This concludes the proof.
\end{subproof}

Intuitively, \Cref{fact:key2} says that, assuming $\Delta^+(G) \leq d$, whether $v$ is full in $G$ is entirely determined by the underlying undirected graph $G^u$. Indeed, observe that property $(2)$ in \Cref{fact:key2} depends only on $G^u$.\footnote{The fact that having no augmenting path, ostensibly a property of a directed graph, is entirely determined by the underlying undirected graph, may appear surprising, but it is a common occurence in combinatorial optimization: see \cite[Chapter~61]{schrijver}, and the discussion in \Cref{sec:related}.} 

\begin{corollary}
\label{corol:key2}
Let $G_0 = (V,E_0)$, $G_1 = (V,E_1)$ be two directed graphs such that $\Delta^+(G_i) \leq d$, and $E_0^u \subseteq E^u_1$. If a vertex $v$ is full in $G_0$, then it is full in $G_1$.
\end{corollary}

\begin{subproof}
This follows directly from the equivalence $(1) \Leftrightarrow (2)$ in \Cref{fact:key2}.
\end{subproof}

We are now ready to prove $\stash \leq \stashopt$.
Let $G = G(n,\varepsilon,d,s)$ be the cuckoo graph defined in \Cref{sec:intro}.
That is, $G = (V,E)$ with $V = T_1 \cup T_2$, and $E = \{\{h_1(x), h_2(x)\} : x \in S\}$.

Recall that algorithm $\A$ starts from an empty graph $G_0 = (V,\varnothing)$. Assume any fixed insertion order $x_1, \dots, x_{|S|}$ on the set $S$ of items.
Let $G_i = (V,E_i)$ denote the graph obtained after algorithm $\A$ has inserted items $x_1, \dots, x_i$. Let $s = |S|$, so that $G_s$ is the graph at the output of algorithm $\A$.
Let $S_\A \subseteq S$ be the set of items sent to the stash by $\A$.
Note that $G^u_s$ is almost the same as $G$, except it is missing the edges corresponding to items in $S_\A$; that is:
\begin{equation}
\label{eq:a}
E = E_s^u \cup \{\{h_1(x),h_2(x)\} : x \in S_\A\}.\tag{$*$}
\end{equation}

By construction of $\A$, whenever an item $x_i \in S_\A$ is sent to the stash, it must be the case that $h_1(x_i)$ and $h_2(x_i)$ are full in $G_{i-1}$. By \Cref{corol:key2}, it follows that $h_1(x_i)$ and $h_2(x_i)$ are also full in $G_s$. Letting $V' = \{h_1(x), h_2(x) : x \in S_\A\}$, the vertices of $V'$ are full in $G_s$. Using \Cref{fact:key2} again, there exists $V'' \supseteq V'$ such that $|E_s[V'']| = d |V''|$. Using (\ref{eq:a}), this yields :
\[
|E[V'']| = |E_s[V'']| + \stash = d |V''| + \stash.
\]
On the other hand, all edges in $E[V'']$ correspond to items that must be stored in $V''$, since both of their endpoints are in $V''$. But the vertices of $V''$ can accomodate at most $d|V''|$ items. It follows that at least $s_\A$ items cannot be stored, and must be sent to the stash. Hence $\stashopt \geq \stash$.
\end{proof}

As a direct corollary, the probability that algorithm $\A$ outputs a solution with a stash of size at most $s$ is equal to the probability that such a solution exists. Hence, the bound from \Cref{prop:up} applies directly to algorithm $A$. For simplicity, we only write the corollary in the case that $s$ and $d$ are constant.

\begin{corollary}\label{final_col}
Les $\stash$ denote the size of the stash after all $n$ items have been inserted using algorithm $\A$. Assume $0 < \varepsilon \leq 0.25$, $d \geq 1 + \ln(1/\varepsilon)/(1-\ln 2)$, and $1 \leq s \leq m/(10e^4)$ are constant. Then $\Pr[\stash > s] = \bigo(n^{-d-s})$.
\end{corollary}

\section{Related work}
\label{sec:related}

Cuckoo hashing was introduced by Pagh and Rodler~\cite{cuckoo}. A variant where each table location can receive $d>1$ items was studied by Dietzfelbinger and Weidling~\cite{balanced}, and a variant with $k > 2$ hash functions was analyzed by Fotakis \emph{et al.}~\cite{fotakis}. Later, Kirsch, Mitzenmacher and Wieder introduced cuckoo hashing with a stash, proving that for all three aforementioned variants of cuckoo hashing, the failure probability decreases exponentially with the stash size~\cite{kmw}. The present work corrects the results of~\cite{kmw} pertaining to the variant with $d > 1$ items per cell. As discussed in \Cref{sec:proof}, our main proof owes much to the analysis in~\cite{balanced}.
In a different direction, Aum\"uller \emph{et al.} have shown that the results of \cite{kmw} regarding cuckoo hashing with a stash can be modified to work with explicit hash function families, rather than uniformly random hash functions~\cite{aumuller}.

Cuckoo hashing relates to balls-and-bins games, and particularly to the so-called two-choice process~\cite{azar,power}. In the simplest version of the two-choice process, $n$ balls are thrown into $m$ bins, one after the other. For each ball, two bins are selected uniformly at random. The ball is inserted into whichever bin holds the fewest balls at insertion time. Cuckoo hashing follows a similar premise, with the difference that whenever a new ball is inserted, the choices made for previous balls can be revisited. That is, each ball can be moved at any time between its two possible locations, rather than this being decided only at insertion time. Our main theorem (\Cref{prop:up}) can be reinterpreted from that perspective: it provides a bound on the probability that the latter variant of the two-choice process admits a solution where the most loaded bin contains at most $d$ items, if we are allowed to ``skip'' inserting at most $s$ balls. The result of Dietzfelbinger and Weidling cited earlier is written from that perspective~\cite{balanced}.

As noted in the introduction, Cuckoo hashing can also be cast as a graph orientability problem. In that light, \Cref{fact:key} and \Cref{thm:opt} are variants of a result due to Hakimi~\cite{hakimi}, which gives a similar equivalent condition for a graph to admit an orientation where the indegree of each vertex $v$ is at least equal to a prescribed value $\ell(v)$. Our own lemmas correspond to the setting where $\ell(v)$ is constant, but we allow for a stash (i.e. at most $s$ edges can be deleted). A result of Frank and Gy\'arf\'as gives a similar condition when one wishes to prescribe both a lower and upper bound on the indegree of each vertex~\cite{frank}. This line of results belongs to the family of so-called ``min-max'' theorems in combinatorial optimization. In particular, they are closely related to linear programming duality. We refer the reader to \cite{schrijver} for an excellent exposition of those relationships.

\subsection*{Acknowledgments}

The authors would like to thank Adam Kirsch, Michael Mitzenmacher, Udi Wieder, as well as Martin Dietzfelbinger, for their helpful comments. This work was supported by the ANR JCJC project SaFED and by the National Science Foundation.

\bibliographystyle{elsarticle-num}
\bibliography{biblio}







\end{document}